\documentclass[10pt]{article}
\usepackage{amsthm,amstext}
\usepackage{amsmath}
\usepackage{graphicx}
\usepackage{amsfonts}
\usepackage{amssymb}
\usepackage{wrapfig}
\theoremstyle{plain}
\usepackage{xy,newlfont,qsymbols,upgreek,pifont}

\usepackage{amsmath,amssymb,amsthm,mathrsfs,amsfonts,dsfont}
\usepackage{amsfonts}
\usepackage{amssymb}
\theoremstyle{plain}
\newtheorem{theorem}{Theorem}[section]

\newtheorem{example}[theorem]{Example}

\newtheorem{proposition}[theorem]{Proposition}

\theoremstyle{definition}
\newtheorem{definition}[theorem]{Definition}

\newtheorem{remark}[theorem]{Remark}

\date{}

\title{On Codes based on $BCK$-algebras}
\author{A. Borumand Saeid, H. Fatemidokht, C. Flaut and M. Kuchaki Rafsanjani}

\begin{document}
\maketitle
\textbf{Abstract.} {\small In this paper, we present some new connections between $BCK$-algebras and binary block codes.}
\begin{equation*}
\end{equation*}

\bigskip

\begin{equation*}
\end{equation*}

\section{Introduction}


$BCI/BCK$-algebras were first introduced in mathematics in 1966 by Y. Imai and
K. Iseki, through the paper [Im, Is; 66], as a generalization of the
concept of set-theoretic difference and propositional calculi. One of the
recent applications of $BCK$-algebras was given in the Coding Theory (see [Fl;
14], [Ju, So; 11]).

\section{Preliminaries}

\begin{definition}
An algebra $(X,\ast ,\theta )$ of type $(2,0)$ is called a \textit{$BCI$-algebra} if the following conditions are
fulfilled:
\begin{itemize}
\item 	$BCI$-1 ~~$((x*y)*(x*z))*(z*y)=\theta$
\item 	$BCI$-2 ~~$(x*(x*y))*y=\theta$
\item 	$BCI$-3 ~~~$x*x=\theta$
\item 	$BCI$-4 ~~~$x*y=\theta$ and $y*x=\theta$ imply $x=y$
\end{itemize}
If a $BCI$-algebra $X$ satisfies the following identity:
\begin{itemize}
\item 	$BCK$-5 ~~$\theta * x=\theta$
\end{itemize}
then $X$ is called a \textit{$BCK$-algebra} [Me, Ju; 94].
\end{definition}

The partial order relation on a $BCI$\textit{/}$BCK$-algebra is defined such that $x\leq y$ if and only if $x\ast y=\theta .$

A $BCI$\textit{/}$BCK$-algebra $X$ is called \textit{commutative }if $x\ast
(x\ast y)=y\ast (y\ast x),$ for all $x,y\in X$ and \textit{implicative }if $%
x\ast (y\ast x)=x,$ for all $x,y\in X.$ 

If $(X,\ast ,\theta )$ and $(Y,\circ ,\theta )$ are two $BCI$\textit{/}%
$BCK$-algebras, a map $f:X\rightarrow Y$ with the property $f\left( x\ast
y\right) =f\left( x\right) \circ f\left( y\right) ,$ for all $x,y\in X,$ is
called a $BCI$\textit{/}$BCK$\textit{-algebras morphism}$.$ If $f$ is a
bijective map, then $f$ is an \textit{isomorphism} of $BCI$\textit{/}%
$BCK$-algebras [Me, Ju; 94].

Hereafter in this paper, $X$ always denotes a finite  $BCI$\textit{/}%
$BCK$-algebra.\medskip

In the following, we will use some notations and results given in the paper [Ju, So; 11].
\begin{definition}
A mapping $\tilde{A}:A \rightarrow X$ is called a $BCK$-function on A, which A and X is a nonempty set and a $BCK$-algebra, respectively.
\end{definition}

\begin{definition}
A cut function of  $\tilde{A}$, for $q \in X$, is defined to be a mapping

$\tilde{A}_q:A \rightarrow \{0,1\}$
\newline
such that

$(\forall x \in A)(\tilde{A}_q(x)=1 \Leftrightarrow q*\tilde{A}(x)=\theta)$
\end{definition}

\begin{definition}
Let $A=\{1,2,\ldots,n\}$ and let X be a $BCK$-algebra. In [Ju, So; 11], to each $BCK$-function $\tilde{A}:A \rightarrow X$ can be associated a binary block-code of length $n$. A codeword in a binary block-code V is $v_x=x_1 x_2 \ldots x_n$  such that $x_i=x_j \Leftrightarrow A_x(i)=j$ for $i \in A$ and $j \in \{0,1\}$.
\end{definition}

Let $v_x=x_1x_2 \ldots x_n$ and $v_y=y_1y_2 \ldots y_n$ be two codewords belonging to a binary block-code V. Define an order relation $\leqslant_c$ on
the set of codewords belonging to a binary block-code V as follows [Ju, So; 11]:

$v_x \leqslant_c v_y \Leftrightarrow y_i \leqslant x_i$ for $i = 1,2,\ldots,n.$


\section{Main results}

\begin{definition}
Let $(S,\leqslant)$ be a partially ordered set. For $q \in S$, we define a mapping

$S_q:S\rightarrow \{0,1\}$
\newline
such that

$(\forall b\in S)(S_q(b)=1 \Leftrightarrow q \leqslant b).$

A codeword $v_x=x_1 x_2 \cdots x_n$ of a binary block-code V is determined as follow:

$x_i=x_j \Leftrightarrow S_x(i)=j$, for $i \in S$ and $j \in \{0,1\}.$
\end{definition}

\begin{example}
Let $S=\{0,1,2,3,4\}$ be a set with a partial order over $S$ showed in the Figure \ref{1}(a).
\begin{figure}[h!]
\centering
\includegraphics [width=0.7\textwidth]{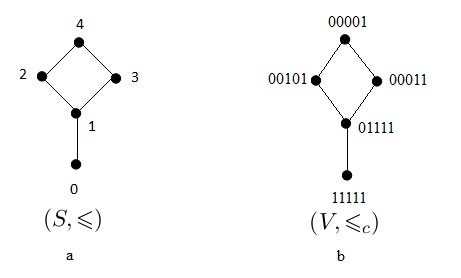}
\caption{a)partial ordering. b)order relation $\leqslant_c$}
\label{1}
\end{figure}
\newpage
then
\newline

\begin{tabular}{l l l l l l}
\hline
$S_s$ & 0 & 1 & 2 & 3 & 4 \\ \hline
$S_0$ & 1 & 1 & 1 & 1 & 1 \\
$S_1$ & 0 & 1 & 1 & 1 & 1 \\
$S_2$ & 0 & 0 & 1 & 0 & 1 \\
$S_3$ & 0 & 0 & 0 & 1 & 1 \\
$S_4$ & 0 & 0 & 0 & 0 & 1 \\
\hline
\end{tabular}
\newline

and thus $V1-P=\{11111,01111,00101,00011,00001\}.$
\end{example}

\begin{example}
Let $S=\{0,1,2,3,4\}$ be a set with a partial order over S  showed in the figure \ref{2}(a).
\begin{figure}[h!]
\centering
\includegraphics [width=0.6\textwidth]{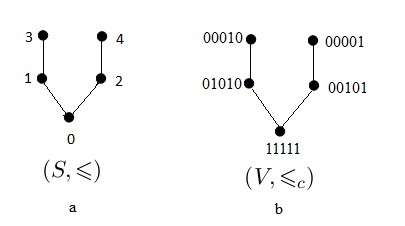}
\caption{a)partial ordering. b)order relation $\leqslant_c$}
\label{2}
\end{figure}
\newpage
then
\newline

\begin{tabular}{l l l l l l}
\hline
$S_s$ & 0 & 1 & 2 & 3 & 4 \\ \hline
$S_0$ & 1 & 1 & 1 & 1 & 1 \\
$S_1$ & 0 & 1 & 0 & 1 & 0 \\
$S_2$ & 0 & 0 & 1 & 0 & 1 \\
$S_3$ & 0 & 0 & 0 & 1 & 0 \\
$S_4$ & 0 & 0 & 0 & 0 & 1 \\
\hline
\end{tabular}
\newline

and thus $V2-P=\{11111,01010,00101,00010,00001\}.$
\end{example}

\begin{example}
Let $S=\{A,B,C,D\}$ be a set with a partial order over S as in the Figure \ref{3}(a).

\begin{figure}[h!]
\centering
\includegraphics [width=0.6\textwidth]{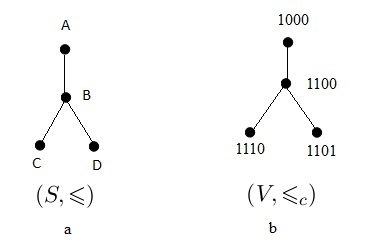}
\caption{a)partial ordering. b)order relation $\leqslant_c$}
\label{3}
\end{figure}

then
\newline

\begin{tabular}{l l l l l}
\hline
$S_s$ & A & B & C & D  \\ \hline
$S_A$ & 1 & 0 & 0 & 0  \\
$S_B$ & 1 & 1 & 0 & 0  \\
$S_C$ & 1 & 1 & 1 & 0  \\
$S_D$ & 1 & 1 & 0 & 1  \\
\hline
\end{tabular}
\newline

therefore $V3-P=\{1000,1100,1110,1101\}.$
\end{example}

In the following, we will compute binary block-code based on Definition 2.4. for $BCK$-algebras. We will show that there is a correspondence between the ordered relation on $BCK$-algebra and partial ordered set.

\begin{example}
Let $X = \{0, 1, 2, 3, 4\}$ be a $BCK$-algebra with the following Cayley table:
\newline

\begin{tabular}{l|l l l l l }
$*$ & 0 & 1 & 2 & 3 & 4 \\ \hline
$0$ & 0 & 0 & 0 & 0 & 0 \\
$1$ & 1 & 0 & 0 & 0 & 0 \\
$2$ & 2 & 1 & 0 & 1 & 0 \\
$3$ & 3 & 3 & 3 & 0 & 0 \\
$4$ & 4 & 4 & 4 & 4 & 0 \\
\end{tabular}


\begin{figure}[h!]
\begin{center}
\includegraphics [width=0.7\textwidth]{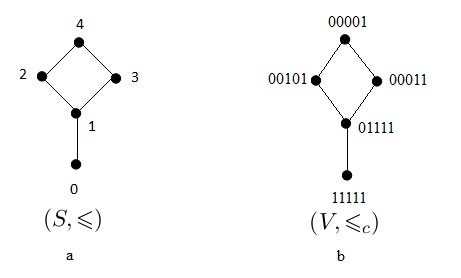}
\end{center}
\caption{a)ordered relation. b)order relation $\leqslant_c$}
\label{4}
\end{figure}
The above figure is the ordered relation on X.

Let $\tilde{A}:X \rightarrow X$ be a $BCK$-function on X given by
\newline

$\tilde{A}=
\begin{pmatrix}
0 & 1 & 2 & 3 & 4 \\
0 & 1 & 2 & 3 & 4
\end{pmatrix}$
\newline

then
\newline

\begin{tabular}{l l l l l l}
\hline
$\tilde{A}_x$ & 0 & 1 & 2 & 3 & 4 \\ \hline
$\tilde{A}_0$ & 1 & 1 & 1 & 1 & 1 \\
$\tilde{A}_1$ & 0 & 1 & 1 & 1 & 1 \\
$\tilde{A}_2$ & 0 & 0 & 1 & 0 & 1 \\
$\tilde{A}_3$ & 0 & 0 & 0 & 1 & 1 \\
$\tilde{A}_4$ & 0 & 0 & 0 & 0 & 1 \\
\hline
\end{tabular}
\newline

thus $V1-B=\{11111,01111,00101,00011,00001\}.$
\end{example}

\begin{example}
Let $X = \{0, 1, 2, 3, 4\}$ be a $BCK$-algebra with the following Cayley table:
\newline

\begin{tabular}{l|l l l l l }
$*$ & 0 & 1 & 2 & 3 & 4 \\ \hline
$0$ & 0 & 0 & 0 & 0 & 0 \\
$1$ & 1 & 0 & 1 & 0 & 1 \\
$2$ & 2 & 2 & 0 & 2 & 0 \\
$3$ & 3 & 1 & 3 & 0 & 3 \\
$4$ & 4 & 4 & 2 & 4 & 0 \\
\end{tabular}
\newline

\begin{figure}[h!]
\begin{center}
\includegraphics [width=0.7\textwidth]{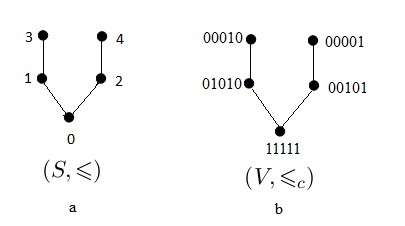}
\end{center}
\caption{a)ordered relation. b)order relation $\leqslant_c$}
\label{5}
\end{figure}

The above figure is the ordered relation on X.

Let $\tilde{A}:X \rightarrow X$ be a $BCK$-function on X given by
\newline

$\tilde{A}=
\begin{pmatrix}
0 & 1 & 2 & 3 & 4 \\
0 & 1 & 2 & 3 & 4
\end{pmatrix}$

then
\newline

\begin{tabular}{l l l l l l}
\hline
$\tilde{A}_x$ & 0 & 1 & 2 & 3 & 4 \\ \hline
$\tilde{A}_0$ & 1 & 1 & 1 & 1 & 1 \\
$\tilde{A}_1$ & 0 & 1 & 0 & 1 & 0 \\
$\tilde{A}_2$ & 0 & 0 & 1 & 0 & 1 \\
$\tilde{A}_3$ & 0 & 0 & 0 & 1 & 0 \\
$\tilde{A}_4$ & 0 & 0 & 0 & 0 & 1 \\
\hline
\end{tabular}
\newline

thus $V2-B=\{11111,01010,00101,00010,00001\}.$
\end{example}

\begin{remark}
 On a partial ordered set with a minimum element $%
\theta $ we can define a $BCK$-algebra structure(see [Fl; 14], $\left(
2.1)\right) $ From the obtained block-codes by the aforesaid methods, it is
obvious that $V1-P$ $=V1-B$ and $V2-P=V2-B$. We think that the problem
occurred because we use only the order of $BCK$-algebra, not its algebraic
properties. From above examples, it is obvious that the method presented in paper [Ju, So; 11] dose not depend on algebraic
properties of $BCK$-algebra. Also the obtained codes are not good codes, since their Hamming
distance is not good. According to the figures 1 to 5, there is a
one-to-one correspondence between the ordering relation $\leqslant $ and
order relation $\leqslant _{c}.$
\end{remark}

\medskip

Let $X$ be a $BCK$-algebra and $V$ be a linear binary block-code with $n$
codewords of length $n.$ We consider the matrix $M_{V}=\left( m_{i,j}\right)
_{i,j\in \{1,2,...,n\}}\in \mathcal{M}_{n}(\{0,1\})$ with the rows
consisting of the codewords of $V.$ This matrix is called \textit{the matrix
associated to the code} $V.\medskip $ We consider the codewords in $V$
lexicographic ordered in the ascending sense. With this remark, for \ $%
V=\{w_{1},....w_{n}\},$ we denote lines in $M_{V}$ with $%
L_{w_{1}},...,L_{w_{n}}.$ Obviously, $w_{1}=\underset{n-\text{time}}{%
\underbrace{00...0}}.$ On $V,$ we define the following multiplication "$\ast$"
\begin{equation}
w_{i}\ast w_{j}=w_{k}\text{ if and only if \ }L_{w_{i}}+L_{w_{j}}=L_{w_{k}}.
\tag{2.1.}
\end{equation}

\begin{proposition}
 \textit{With this multiplication,} $\left(
V,\ast ,\theta \right) ,$ \textit{where} $\theta =w_{1},$ \textit{becomes an
abelian group.}$\Box \medskip $
\end{proposition}

\begin{remark}
The above group is a $BCI$-algebra.\medskip
\end{remark}

\begin{example}
We consider the binary linear code $%
C=\{0000,0001,0010,0011\}=\{\theta ,A,B,C\}.$ The associated
$BCI$-algebra(group) is $X=\{\theta ,A,B,C\}$ with zero element $\theta $ and
multiplication given in the following table:

\begin{tabular}{l|llll}
$\ast $ & $\theta $ & $A$ & $B$ & $C$ \\ \hline
$\theta $ & $\theta $ & \multicolumn{1}{|l}{$A$} & \multicolumn{1}{|l}{$B$}
& \multicolumn{1}{|l|}{$C$} \\ \cline{2-5}
$A$ & $A$ & \multicolumn{1}{|l}{$\theta $} & \multicolumn{1}{|l}{$C$} &
\multicolumn{1}{|l|}{$B$} \\ \cline{2-5}
$B$ & $B$ & \multicolumn{1}{|l}{$C$} & \multicolumn{1}{|l}{$\theta $} &
\multicolumn{1}{|l|}{$A$} \\ \cline{2-5}
$C$ & $C$ & \multicolumn{1}{|l}{$B$} & \multicolumn{1}{|l}{$A$} &
\multicolumn{1}{|l|}{$\theta $} \\ \cline{2-5}
\end{tabular}
\end{example}

\ \ \ \

\begin{definition}
Let $(X,\ast ,\theta )$ be a $BCI$/$BCK$-algebra,
and $I\subseteq X.$ We say that $I$ is a \textit{right-ideal} if $\theta \in
I$ and $x\in I,y\in X$ imply $x\ast y\in I$. An ideal $I$ of a
$BCI$/$BCK$-algebra $X$ is called \textit{a closed ideal }if it is also \textit{a%
} \textit{subalgebra }of $X$ (i.e. $\theta \in I$ and if $x,y\in I$ it
results that $x\ast y\in I$).\medskip
\end{definition}

Let $C$ be a binary block code. In Theorem 2.9, from [Fl; 14], we find a
$BCK$-algebra $X$ such that the obtained binary block-code $V_{X}$ contains
the binary block-code $C$ as a subset.

Let $C$ be a binary block code with $m$ codewords of length $q.$ With the
above notations, let $X$ \ be the associated $BCK$-algebra and \ $W$\ $%
=\{\theta ,w_{1},...,w_{m+q}\}~$the associated binary block code which
include the code $C.$ We consider the codewords $\theta
,w_{1,}w_{2},...,w_{m+q}$ lexicographic ordered, $\theta \geq
_{lex}w_{1}\geq _{lex}w_{2}\geq _{lex}...\geq _{lex}w_{m+q}.$ Let $M\in
\mathcal{M}_{m+q+1}(\{0,1\})$ be the associated matrix with the rows $\theta
,w_{1},...,w_{m+q},$ in this order. We denote with $L_{w_{i}}$ and $C_{w_{j}}
$ the lines and columns in the matrix $M$. The sub-matrix $M^{\prime }$ of
the matrix $M$ with the rows $L_{w_{1}},...,L_{w_{m}}$ and the columns $%
C_{w_{m+1}},...,C_{w_{m+q}}-~$is the matrix associated to the code $C.$
\medskip

\begin{proposition}
\textit{With the above notations, we have that }$%
\{\theta ,w_{m+1},...,w_{m+q}\}$\textit{\ determines a closed  right ideal
in the algebra} $X.\medskip $
\end{proposition}

\begin{proof}
Let $Y=$ $\{\theta ,w_{m+1},...,w_{m+q}\}.$ Due to the
multiplications and the order relation $\preceq $ given by the relations $%
\left( 2.1\right) $ and $\left( 1.1\right) $ from [Fl; 14], we can have only
the following two possibilities:  $w_{i}\ast w_{j}=\theta $ or $w_{i}\ast
w_{j}=w_{i}.$  Therefore $Y$ is a right-ideal in $X.$ The multiplication $%
\left( 2.1.\right) $ is :%
\begin{equation*}
\left\{
\begin{array}{c}
\theta \ast x=\theta \text{ and }x\ast x=\theta ,\forall x\in X; \\
x\ast y=\theta ,\text{ if \ }x\leq y,\ \ \ x,y\in X; \\
x\ast y=x,\text{ otherwise.}%
\end{array}%
\right.
\end{equation*}%
$\Box \medskip $
\end{proof}

\begin{remark}
 From Proposition 3.12, we obtain that to each binary
block code we can associate a $BCK$-algebra in which this code determines a
right ideal.\medskip
\end{remark}

Let $A$ be a nonempty set and $X$ \ be a $BCK$-algebra.\medskip

\begin{proposition}
\textit{Let} $C$ \textit{be a binary block code
with} $m$ \textit{codewords of length} $q$ \textit{and let} $X$ \ \textit{be
the associated $BCK$-algebra, as the above. Therefore, there are the sets} $A$
\textit{and} $B\subseteq X,$ t\textit{he $BCK$-function} $f:A\rightarrow X$
\textit{and a cut function} $\ f_{r}$ \textit{such that}
\begin{equation*}
C=\{f_{r}:A\rightarrow \{0,1\}~/~\ f_{r}\left( x\right) =1,\text{if \ and \
only \ if \ }r\ast \ f\left( x\right) =\theta ,\forall x\in A,r\in B\}.\Box
\end{equation*}
\end{proposition}

\begin{remark}
i) Let $S=\{1,2,...,n\}$ be the set with $n$ elements.
We know that $\left( \mathcal{P}\left( S\right) ,\Delta ,\cap \right) $ is a
Boolean ring, where $\mathcal{P}\left( S\right) $ is the power set of the
set $S,\Delta $ is symmetric difference of the sets and $\cap $ is the
intersection of two sets. Let $\mathfrak{F}=\{f:S\rightarrow \{0,1\}~/$ $f$
\ function$\}.$ To each $f\in \mathfrak{F}$ corresponds a binary block
codeword. To \ each binary block codeword $c_{1}$ corresponds an element
from $\mathcal{P}\left( S\right) $. Indeed, \ to each binary codeword $%
c=(i_{1},...,i_{n})$ we will associate the set $I_{c}=%
\{j_{1},j_{2},...,j_{k}\}\in \mathcal{P}\left( S\right) $ such that $%
i_{j_{1}}=i_{j_{2}}=...=i_{j_{k}}=1.$

ii) Using the above established correspondence, if \ $C=%
\{c_{1},c_{2},...,c_{m}\}$ is a linear binary block code and $%
Q=\{I_{c_{1}},I_{c_{2}},...,I_{c_{m}}\}\subseteq \mathcal{P}\left( S\right)
, $ where $I_{c_{i}}~$is the associated subset for the codeword $c_{1}$,
then $Q$ is a sub-ring in the Boolean ring $\left( \mathcal{P}\left( S\right)
,\Delta ,\cap \right) .$ It results a bijective map between the sub-rings of
the Boolean ring $\left( \mathcal{P}\left( S\right) ,\Delta ,\cap \right) $
and linear binary block codes with codewords of length $n.\medskip $
\end{remark}

\begin{example}
 i) Let $C=\{0000,0001,0010,0011\}=%
\{w_{6},w_{7},w_{8},w_{9}\}$ be a linear binary block code and let $%
X=\{\theta ,w_{2},w_{3},w_{4},w_{5},w_{6},w_{7},w_{8},w_{9}\}$ be the
obtained $BCK$-algebra as in Theorem 2.9 from [Fl; 14]. $\ $The multiplication
of this algebra is given in the below table

\begin{tabular}{l|lllllllll}
$\ast $ & $\theta $ & $w_{2}$ & $w_{3}$ & $w_{4}$ & $w_{5}$ & $w_{6}$ & $%
w_{7}$ & $w_{8}$ & $w_{9}$ \\ \hline
$\theta $ & $\theta $ & \multicolumn{1}{|l}{$\theta $} & \multicolumn{1}{|l}{%
$\theta $} & \multicolumn{1}{|l}{$\theta $} & \multicolumn{1}{|l}{$\theta $}
& \multicolumn{1}{|l}{$\theta $} & \multicolumn{1}{|l}{$\theta $} &
\multicolumn{1}{|l}{$\theta $} & \multicolumn{1}{|l|}{$\theta $} \\
\cline{2-10}
$w_{2}$ & $w_{2}$ & \multicolumn{1}{|l}{$\theta $} & \multicolumn{1}{|l}{$%
w_{2}$} & \multicolumn{1}{|l}{$w_{2}$} & \multicolumn{1}{|l}{$w_{2}$} &
\multicolumn{1}{|l}{$w_{2}$} & \multicolumn{1}{|l}{$w_{2}$} &
\multicolumn{1}{|l}{$\theta $} & \multicolumn{1}{|l|}{$\theta $} \\
\cline{2-10}
$w_{3}$ & $w_{3}$ & \multicolumn{1}{|l}{$w_{3}$} & \multicolumn{1}{|l}{$%
\theta $} & \multicolumn{1}{|l}{$w_{3}$} & \multicolumn{1}{|l}{$w_{3}$} &
\multicolumn{1}{|l}{$w_{3}$} & \multicolumn{1}{|l}{$w_{3}$} &
\multicolumn{1}{|l}{$\theta $} & \multicolumn{1}{|l|}{$w_{3}$} \\
\cline{2-10}
$w_{4}$ & $w_{4}$ & \multicolumn{1}{|l}{$w_{4}$} & \multicolumn{1}{|l}{$w_{4}
$} & \multicolumn{1}{|l}{$\theta $} & \multicolumn{1}{|l}{$w_{4}$} &
\multicolumn{1}{|l}{$w_{4}$} & \multicolumn{1}{|l}{$w_{4}$} &
\multicolumn{1}{|l}{$w_{4}$} & \multicolumn{1}{|l|}{$\theta $} \\
\cline{2-10}
$w_{5}$ & $w_{5}$ & \multicolumn{1}{|l}{$w_{5}$} & \multicolumn{1}{|l}{$w_{5}
$} & \multicolumn{1}{|l}{$w_{5}$} & \multicolumn{1}{|l}{$\theta $} &
\multicolumn{1}{|l}{$w_{5}$} & \multicolumn{1}{|l}{$w_{5}$} &
\multicolumn{1}{|l}{$w_{5}$} & \multicolumn{1}{|l|}{$w_{5}$} \\ \cline{2-10}
$w_{6}$ & $w_{6}$ & \multicolumn{1}{|l}{$w_{6}$} & \multicolumn{1}{|l}{$w_{6}
$} & \multicolumn{1}{|l}{$w_{6}$} & \multicolumn{1}{|l}{$w_{6}$} &
\multicolumn{1}{|l}{$\theta $} & \multicolumn{1}{|l}{$w_{6}$} &
\multicolumn{1}{|l}{$w_{6}$} & \multicolumn{1}{|l|}{$w_{6}$} \\ \cline{2-10}
$w_{7}$ & $w_{7}$ & \multicolumn{1}{|l}{$w_{7}$} & \multicolumn{1}{|l}{$w_{7}
$} & \multicolumn{1}{|l}{$w_{7}$} & \multicolumn{1}{|l}{$w_{7}$} &
\multicolumn{1}{|l}{$w_{7}$} & \multicolumn{1}{|l}{$\theta $} &
\multicolumn{1}{|l}{$w_{7}$} & \multicolumn{1}{|l|}{$w_{7}$} \\ \cline{2-10}
$w_{8}$ & $w_{8}$ & \multicolumn{1}{|l}{$w_{8}$} & \multicolumn{1}{|l}{$w_{8}
$} & \multicolumn{1}{|l}{$w_{8}$} & \multicolumn{1}{|l}{$w_{8}$} &
\multicolumn{1}{|l}{$w_{8}$} & \multicolumn{1}{|l}{$w_{8}$} &
\multicolumn{1}{|l}{$\theta $} & \multicolumn{1}{|l|}{$w_{8}$} \\
\cline{2-10}
$w_{9}$ & $w_{9}$ & \multicolumn{1}{|l}{$w_{9}$} & \multicolumn{1}{|l}{$w_{9}
$} & \multicolumn{1}{|l}{$w_{9}$} & \multicolumn{1}{|l}{$w_{9}$} &
\multicolumn{1}{|l}{$w_{9}$} & \multicolumn{1}{|l}{$w_{9}$} &
\multicolumn{1}{|l}{$w_{9}$} & \multicolumn{1}{|l|}{$\theta $} \\
\cline{2-10}
\end{tabular}

From Proposition 3.12, we \ remark that $\{\theta ,w_{6},w_{7},w_{8},w_{9}\}$
is a right ideal in the $BCK$-algebra $X.$ From \ Proposition 3.14, for $%
A=\{w_{6},w_{7},w_{8},w_{9}\}$ and $B=\{w_{2},w_{3},w_{4},w_{5}\},$ we
recover the initial code $C.\medskip $
\end{example}

\begin{example}
For the same linear binary block code\textbf{\ }$%
C=\{0000,0001,0010,0011\},$ let $Q=\{\varnothing ,\{4\},\{3\},\{3,4\}\}$ as
in Remark 3.15 ii). It is clear that $Q$ is a sub-ring in the Boolean ring $%
\left( \mathcal{P}\left( \{1,2,3,4\}\right) ,\Delta ,\cap \right) $ and $C$
can be considered as a sub-ring of this Boolean ring.\medskip
\end{example}

\begin{remark}
In [Fl; 14], Theorem 2.2, \ the studied binary block
codes have Hamming distance equal with $1.$ In the same paper, Theorem 2.9,
to an arbitrary binary block code $C$ we associate a $BCK$ algebra $X$ and the
code associated to this algebra includes the code $C.~$Proposition 3.14
improved this theorem since we can even obtain the code $C$ and from
Proposition 3.12\ we have that the code $C$ generate a right ideal in the
algebra $X.$%
\end{remark}

\begin{remark}
The  obtained results of above remarks and propositions can be illustrated by partially ordered sets. Let $C$ be a binary block code with $m$ codewords of length $q$. According to Proposition 2.8 and Theorem 2.9 in [Fl; 14], we can find the matrix $M\in \mathcal{M}_{m+q+1}(\{0,1\})$ that is the matrix associated to the code $C$. Let $S$ be the associated partially ordered set. Therefore, there are the sets $A$ and $B \subseteq S$ and the function $f:A\rightarrow S$, such that we can define the bellow set:
\begin{equation*}
C=\{f_{r}:A\rightarrow \{0,1\}~/~\ f_{r}\left( b\right) =1,\text{if \ and \
only \ if \ }r \leq b ,\forall b\in A,r\in B\}.
\end{equation*}
Here, $A = \{m+2, \cdots , m+q+1\}$ and $B = \{2, \cdots, m+1\}$.
\end{remark}

\begin{example}
Let $C=\{0000,0001,0010,0011\}$ be a linear binary block code and let $S=\{1,2,3,4,5,6,7,8,9\}$. In this example $m=q=4$. The matrix associated to the code C is:
\newline

\begin{tabular}{|l|l|l|l|l|l|l|l|l|}
\hline
1 & 1 & 1 & 1 & 1 & 1 & 1 & 1 & 1 \\ \hline
0 & 1 & 0 & 0 & 0 & 0 & 0 & 1 & 1 \\ \hline
0 & 0 & 1 & 0 & 0 & 0 & 0 & 1 & 0 \\ \hline
0 & 0 & 0 & 1 & 0 & 0 & 0 & 0 & 1 \\ \hline
0 & 0 & 0 & 0 & 1 & 0 & 0 & 0 & 0 \\ \hline
0 & 0 & 0 & 0 & 0 & 1 & 0 & 0 & 0 \\ \hline
0 & 0 & 0 & 0 & 0 & 0 & 1 & 0 & 0 \\ \hline
0 & 0 & 0 & 0 & 0 & 0 & 0 & 1 & 0 \\ \hline
0 & 0 & 0 & 0 & 0 & 0 & 0 & 0 & 1 \\
\hline
\end{tabular}
\newline
\begin{figure}[h!]
\begin{center}
\includegraphics [width=0.5\textwidth]{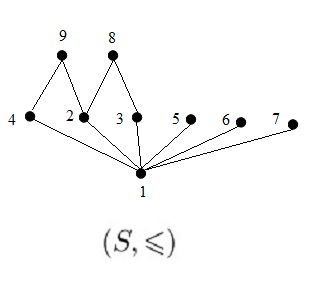}
\end{center}
\caption{partial ordering.}
\label{6}
\end{figure}

The above figure is partial ordering over S.
From above Proposition, $A=\{6,7,8,9\}$ and $B=\{2,3,4,5\}$ that from A and B, we can recover the initial code C.

\end{example}
\[\]
\textbf{Conclusions.} Even if, from the above examples, appears that the associated binary block codes depend only from the order relation defined on a $BCK$-algebra, will be very interesting to study in a further paper how and if the properties of $BCK$-algebras can influence the properties of the associated binary block codes.

\[\]
\textbf{References}%
\begin{equation*}
\end{equation*}

[Im, Is; 66] Y. Imai, K. Iseki, \textit{On axiom systems of propositional
calculi}, Proc. Japan Academic, \textbf{42(1966)}, 19-22.

[Fl; 14] C. Flaut, \textit{$BCK$-algebras arising from block codes, }arxiv.

[Ju, So; 11] Y. B. Jun, S. Z. Song, \textit{Codes based on $BCK$-algebras},
Inform. Sciences., \textbf{181(2011)}, 5102-5109.

[Me, Ju; 94] J. Meng and Y.B. Jun, \textit{$BCK$-algebras}, Kyung Moon Sa Co. Seoul, Korea, 1994.

Arsham Borumand Saeid\\
Dept. of Math. Shahid Bahonar University of Kerman, Kerman, Iran\\e-mail: arsham@uk.ac.ir\\
Cristina Flaut\\Faculty of Mathematics and Computer Science,
Ovidius University,
Bd. Mamaia 124, 900527, Constanta,
ROMANIA\\e-mail: cristina$\_$flaut@yahoo.com\\H. Fatemidokht and Marjan Kuchaki Rafsanjani\\Dept. of Computer Science Shahid Bahonar University of Kerman, Kerman, Iran\\e-mail: kuchaki@uk.ac.ir.
\end{document}